\documentclass[11pt,a4paper]{article}
\usepackage[T1]{fontenc}
\usepackage[utf8]{inputenc}
\usepackage{lmodern}
\usepackage[margin=25mm]{geometry}
\usepackage{amsmath,amssymb,amsthm,bm}
\usepackage{booktabs,tabularx,array}
\usepackage[numbers,sort&compress]{natbib}
\usepackage{microtype}
\usepackage{xcolor}
\usepackage{hyperref}
\hypersetup{colorlinks=true,linkcolor=blue!45!black,citecolor=blue!45!black,
 urlcolor=blue!45!black,pdftitle={Geometry as Thermodynamics: Entropy Stationarity and Horizon Residues},
 pdfauthor={Wen-Xiang Chen}}
\numberwithin{equation}{section}
\newtheorem{proposition}{Proposition}[section]
\newtheorem{lemma}[proposition]{Lemma}
\newcommand{\dd}{\mathrm{d}}
\newcommand{\ii}{\mathrm{i}}
\newcommand{\Res}{\operatorname*{Res}}
\newcommand{\SW}{S_{\mathrm{W}}}
\newcommand{\Htemp}{T_{\mathrm{H}}}
\newcommand{\QH}{\mathcal{Q}_{\mathrm{H}}}
\newcommand{\MK}{M_{\mathrm{K}}}
\newcommand{\cS}{\mathcal{S}}
\newcommand{\cE}{\mathcal{E}}
\newcommand{\eps}{\varepsilon}
\title{\bfseries Geometry as Thermodynamics:\\[2mm]
Entropy Stationarity and Horizon Residues}
\author{Wen-Xiang Chen\thanks{Email: \href{mailto:wxchen4277@qq.com}{wxchen4277@qq.com}.}\\[1mm]
\small School of Electronic Information,\\
\small Guangzhou City University of Technology, Guangzhou, China}
\date{}

\begin{document}
\maketitle
\begin{abstract}
Thermodynamic descriptions of gravity involve both variational conditions for spacetime dynamics and analytic structures associated with horizons. We examine their relation while keeping their assumptions distinct. First, we present an explicit derivation of the Einstein equation from the established null-vector entropy functional, including the null constraint, the matter term, and the integration constant associated with the cosmological constant. Second, for an analytic static spherical geometry with a nondegenerate Killing horizon, we define a meromorphic radial one-form whose residue is the inverse of twice the surface gravity. Euclidean regularity then fixes the Hawking temperature. Combining this residue with the Einstein--Hilbert Noether charge gives a normalized contour representation of the Wald entropy. The construction reproduces the Schwarzschild temperature, entropy, and Smarr relation, but does not constitute an independent microscopic derivation of the area law. We establish the limits of a stronger identification between pole structure and dynamics: an asymptotically flat family can retain the Schwarzschild horizon residue, area, surface gravity, and mass while violating the vacuum Einstein equation, and an extremal charged solution possesses a higher-order pole. We also show why exponentiating an unspecified entropy functional does not produce a universal entropy residue. The resulting framework separates entropy stationarity, horizon analyticity, and charge normalization, providing explicit consistency tests for further thermodynamic interpretations of gravitational singularities.
\end{abstract}
\noindent\textbf{Keywords:} gravitational thermodynamics; entropy variational principle; Hawking temperature; complex analysis; Noether charge

\section{Introduction}\label{sec:intro}
Black-hole mechanics, horizon entropy, and quantum particle creation establish a precise connection between spacetime geometry and thermodynamics \cite{Bardeen1973,Bekenstein1973,Hawking1975}. The Unruh effect further associates a temperature with acceleration even in flat spacetime \cite{Unruh1976}. These results motivate the question of whether gravitational dynamics can be formulated as a thermodynamic consistency condition, and whether complex analysis can make the associated horizon data easier to identify.

Several established approaches address different parts of this question. Jacobson derived the Einstein equation from a local Clausius relation and an entropy proportional to area \cite{Jacobson1995}. Padmanabhan and Paranjape obtained gravitational field equations by extremizing a functional of null vectors while keeping the metric fixed \cite{PadmanabhanParanjape2007}. Wald and Iyer related the entropy of a stationary black hole to a diffeomorphism Noether charge \cite{Wald1993,IyerWald1994}. The Euclidean approach fixes the thermal period by regularity \cite{GibbonsHawking1977}. These constructions do not have identical inputs, and agreement between their results should not be confused with a derivation of all of their inputs from a single postulate.

The thermodynamic interpretation is also broader than any particular entropic-force proposal. For example, Verlinde's construction \cite{Verlinde2011} addresses forces through changes of information on screens, whereas the null-vector variational principle constrains the background field equations. A review of thermodynamic approaches is given in Ref.~\cite{Padmanabhan2010}. Topological methods for calculating Hawking temperature, including the Robson--Villari--Biancalana construction \cite{Robson2019}, provide a further reason to examine what is encoded in analytic horizon data.

The present work connects three precisely specified objects: an entropy functional whose stationarity yields gravitational dynamics, a meromorphic one-form that encodes a nondegenerate horizon temperature, and a Noether charge that supplies the entropy normalization. Its specific contribution is a consistency analysis of this connection, including a normalized contour representation and counterexamples to an equivalence between pole order and the Einstein equation. The underlying entropy variational principle and the standard black-hole thermodynamic relations are established results and are not claimed as new.

A central distinction is between a spacetime singularity and a singularity of an analytically continued auxiliary quantity. The former may indicate geodesic incompleteness under the hypotheses of a singularity theorem \cite{Penrose1965}; the latter depends on which function or differential is continued. A regular nonextremal horizon can generate a pole in a radial propagation differential without being a curvature singularity. This distinction is essential if singularities are to motivate, rather than merely label, a thermodynamic explanation.

Section~\ref{sec:variation} derives the field equation. Section~\ref{sec:exponential} examines exponentiated entropy. Sections~\ref{sec:temperature} and \ref{sec:charges} develop the residue and charge construction. Section~\ref{sec:tests} gives explicit checks and counterexamples. Section~\ref{sec:discussion} states the physical interpretation and its limitations.

\section{Entropy stationarity and the Einstein equation}\label{sec:variation}
\subsection{Conventions and variational domain}
We work in four dimensions, with signature $(-,+,+,+)$ and $c=\hbar=k_{\mathrm B}=1$, retaining Newton's constant $G$. The connection is Levi--Civita and our curvature convention is
\begin{equation}
 [\nabla_c,\nabla_d]v^a=R^a{}_{bcd}v^b,
 \qquad R_{bd}=R^a{}_{bad}.
\end{equation}
The stress tensor is symmetric and covariantly conserved, $\nabla^a T_{ab}=0$. Let $V$ be a finite region on which the fields are sufficiently differentiable for the integrations by parts below. We vary an auxiliary, dimensionless vector $n^a$ at fixed metric and stress tensor, subject to $n^a n_a=0$. Variations have compact support in the interior of $V$. This removes boundary variations without setting the value of a boundary integral to zero.

The vector $n^a$ is not required to be a Killing field. It parametrizes arbitrary local null directions; the field equation follows from imposing stationarity for every such direction. A horizon Killing field $\chi^a$ will be introduced separately in Section~\ref{sec:temperature}. Smoothness suffices here; local real analyticity is an additional assumption only in the complex analysis.

\subsection{Why an isolated derivative-square term is insufficient}
Consider the trial functional
\begin{equation}
 \cS_0[n]=\frac{1}{8\pi G}\int_V\dd^4x\sqrt{-g}\,
 \nabla_a n_b\nabla^a n^b.
 \label{eq:trial}
\end{equation}
At fixed metric, its unconstrained variation is
\begin{equation}
 \delta\cS_0=-\frac{1}{4\pi G}\int_V\dd^4x\sqrt{-g}\,
 (\Box n_b)\delta n^b,
 \label{eq:trialvariation}
\end{equation}
up to the boundary variation. Thus the Euler--Lagrange equation is a vector wave equation, not the Einstein equation. A null constraint adds a multiplier times $n_b$ but does not remove this differential operator.

For an exact Killing field one has $\Box\chi^a=-R^a{}_b\chi^b$ in our convention. This identity cannot be imposed on arbitrary variations of $n^a$. A generic curved spacetime does not admit an exact Killing field in every null direction. Nor does an approximate local symmetry justify replacing the full Euler--Lagrange operator by a curvature contraction without controlling the approximation. The tensor $g^{ac}g^{bd}$ alone also lacks the antisymmetries of a Riemann-type tensor. These points require a different contraction in the functional.

\subsection{The antisymmetrized functional}
For Einstein gravity, introduce
\begin{equation}
 P_{ab}{}^{cd}=\frac{1}{32\pi G}
 (\delta_a^c\delta_b^d-\delta_a^d\delta_b^c),
 \qquad \nabla_eP_{ab}{}^{cd}=0.
 \label{eq:P}
\end{equation}
With indices lowered, this tensor has the algebraic symmetries needed in the null-vector entropy principle \cite{PadmanabhanParanjape2007}. The functional is
\begin{align}
 \cS[n]&=\int_V\dd^4x\sqrt{-g}
 \left(4P_{ab}{}^{cd}\nabla_c n^a\nabla_d n^b-T_{ab}n^a n^b\right)
 \nonumber\\
 &=\int_V\dd^4x\sqrt{-g}\left\{
 \frac{1}{8\pi G}\left[(\nabla_a n^a)^2-\nabla_a n^b\nabla_b n^a\right]
 -T_{ab}n^a n^b\right\}.
 \label{eq:functional}
\end{align}
In these conventions $[G]=L^2$, $[T_{ab}]=L^{-4}$, and $\cS$ is dimensionless. The relative coefficient of the two derivative terms is essential. The coupling $G$ is an input, calibrated to gravitational physics; stationarity does not determine its numerical value.

To exhibit the curvature term directly, define
\begin{equation}
 J^a=n^a\nabla_b n^b-n^b\nabla_b n^a.
\end{equation}
The product rule and the curvature commutator give
\begin{align}
 \nabla_a J^a
 &= (\nabla_a n^a)^2-\nabla_a n^b\nabla_b n^a
 +n^b(\nabla_b\nabla_a n^a-\nabla_a\nabla_b n^a)
 \nonumber\\
 &= (\nabla_a n^a)^2-\nabla_a n^b\nabla_b n^a-R_{ab}n^a n^b.
 \label{eq:identity}
\end{align}
Consequently,
\begin{equation}
 \cS[n]=\frac{1}{8\pi G}\int_{\partial V}J^a\dd\Sigma_a
 +\int_V\dd^4x\sqrt{-g}\,
 \left(\frac{R_{ab}}{8\pi G}-T_{ab}\right)n^a n^b.
 \label{eq:boundarybulk}
\end{equation}
The cancellation of second derivatives of the freely varied vector into a curvature commutator is the mechanism missing from Eq.~\eqref{eq:trial}.

\subsection{Null constraint and cosmological constant}
Define $\cE_{ab}=R_{ab}/(8\pi G)-T_{ab}$ and enforce the constraint using
\begin{equation}
 \cS_\lambda[n]=\cS[n]+
 \int_V\dd^4x\sqrt{-g}\,\lambda(x)g_{ab}n^a n^b.
\end{equation}
Independent variation of $n^a$ and $\lambda$ yields
\begin{equation}
 (\cE_{ab}+\lambda g_{ab})n^b=0,
 \qquad n^a n_a=0.
 \label{eq:EL}
\end{equation}
Contracting the first equation with $n^a$ and requiring stationarity for every null direction gives
\begin{equation}
 (R_{ab}-8\pi GT_{ab})n^a n^b=0
 \quad\text{for every null } n^a.
 \label{eq:nullprojection}
\end{equation}

\begin{lemma}\label{lem:null}
If a symmetric tensor $X_{ab}$ on a Lorentzian tangent space satisfies $X_{ab}k^a k^b=0$ for every null vector $k^a$, then $X_{ab}=\phi g_{ab}$ for some scalar $\phi$.
\end{lemma}
\begin{proof}
Choose an orthonormal frame and let $k^a=(1,\bm u)$ with $\bm u\cdot\bm u=1$. Comparing the equations for $\bm u$ and $-\bm u$ gives $X_{0i}u^i=0$, hence $X_{0i}=0$. The remaining equation is $X_{ij}u^i u^j=-X_{00}$ for every unit vector. It implies $X_{ij}=-X_{00}\delta_{ij}$. Thus $X_{ab}=\phi\eta_{ab}$ with $\phi=-X_{00}$, and the result is tensorial.
\end{proof}

The lemma implies
\begin{equation}
 R_{ab}-8\pi GT_{ab}=\phi g_{ab}.
 \label{eq:phi}
\end{equation}
Taking a divergence, using stress conservation and the contracted Bianchi identity, gives
\begin{equation}
 \nabla_b\phi=\frac12\nabla_bR,
 \qquad \phi=\frac12 R-\Lambda,
 \label{eq:bianchi}
\end{equation}
where $\Lambda$ is constant on a connected region. Substitution yields
\begin{equation}
 \boxed{R_{ab}-\frac12Rg_{ab}+\Lambda g_{ab}=8\pi GT_{ab}.}
 \label{eq:einstein}
\end{equation}
In particular, $\phi$ itself need not be constant in the presence of matter. In vacuum the trace gives $R=4\Lambda$, so $R_{ab}=\Lambda g_{ab}$; this does not require the Weyl tensor to vanish or the spacetime to have constant sectional curvature.

Conversely, Eq.~\eqref{eq:einstein} implies $\cE_{ab}=(R/2-\Lambda)g_{ab}/(8\pi G)$, so Eq.~\eqref{eq:EL} holds for any null $n^a$ with the corresponding multiplier. This establishes the equivalence within the stated variational assumptions.

\subsection{Meaning of stationarity}
The matter contraction $T_{ab}n^a n^b$ is generally nonzero for null $n^a$; only $g_{ab}n^a n^b$ vanishes. It represents a null energy-flux term in the entropy construction, rather than an independently derived microscopic matter entropy.

Equation~\eqref{eq:boundarybulk} becomes a boundary functional on shell because its bulk term is proportional to $n^a n_a$. This does not identify its value on an arbitrary region with a black-hole entropy without further boundary and normalization choices. Moreover, stationarity does not establish a strict maximum or a stability theorem. Indeed, on shell, compactly supported changes between null fields with identical boundary data leave the bulk contribution zero. Thermodynamic stability requires a physical ensemble and variations of the corresponding states, not merely the auxiliary-vector variation used above.

\section{Analytic structure of exponentiated entropy}\label{sec:exponential}
A spacetime functional $\cS[g,n]$ is a number once its fields and integration domain have been specified. To obtain a complex function one must additionally prescribe a family of fields or domains parametrized by a complex coordinate $z$. Even if this produces a well-defined $\cS(z)$, the choice
\begin{equation}
 \mathcal{Z}(z)=\exp[\cS(z)]
 \label{eq:exponential}
\end{equation}
does not by itself identify either a statistical partition function or a charge-generating differential.

Three elementary cases expose the difficulty. If $\cS(z)$ is holomorphic at $z=0$, then $\mathcal{Z}(z)$ is holomorphic and nonzero there, so $\Res_{0}[\mathcal{Z}(z)\dd z]=0$. If $\cS$ has a pole, its exponential has an essential singularity; for example,
\begin{equation}
 e^{c/z}=\sum_{j=0}^{\infty}\frac{c^j}{j!}z^{-j},\qquad c\ne0.
\end{equation}
A pole of $\mathcal{Z}$ can instead arise from a logarithm,
\begin{equation}
 \cS(z)=-p\log(z/\ell)+s_{\mathrm{reg}}(z),
 \qquad \mathcal{Z}(z)=\left(\frac{\ell}{z}\right)^p e^{s_{\mathrm{reg}}(z)},
 \label{eq:logentropy}
\end{equation}
where $\ell$ makes the logarithm dimensionless. A positive integer $p$ produces a single-valued pole of order $p$; noninteger $p$ generally introduces a branch point. Neither the existence of this logarithm nor the choice $p=1$ follows from entropy stationarity.

There is a separate normalization obstruction. An additive shift $\cS\mapsto\cS+C$ preserves the variational equation but rescales any residue of $\mathcal{Z}\dd z$ by $e^C$. Its proposed equality to a mass or an entropy is therefore not fixed by stationarity. In a canonical ensemble, the partition function instead has the form
\begin{equation}
 Z(\beta)=\operatorname{Tr}e^{-\beta H},
 \qquad S=\log Z+\beta\langle H\rangle,
\end{equation}
and a semiclassical gravitational contribution involves $e^{-I_E}$ with specified boundary conditions \cite{GibbonsHawking1977}. An exponential of entropy alone supplies none of this ensemble structure.

\section{A meromorphic differential for horizon temperature}\label{sec:temperature}
\subsection{Static spherical geometry and Euclidean regularity}
Consider the static metric
\begin{equation}
 \dd s^2=-e^{2\psi(r)}F(r)\dd t^2+\frac{\dd r^2}{F(r)}+r^2\dd\Omega^2.
 \label{eq:metric}
\end{equation}
Let $r=r_h>0$ be an outer nondegenerate Killing horizon, with
\begin{equation}
 F(r_h)=0,\qquad F'(r_h)>0,\qquad
 |\psi(r_h)|<\infty.
 \label{eq:horizon}
\end{equation}
Assume that $F$ and $\psi$ are real analytic near $r_h$. The Killing field is $\chi^a=(\partial_t)^a$, with its normalization fixed once and for all. In an asymptotically flat example we choose $g_{tt}\to-1$ at infinity. Define
\begin{equation}
 \kappa=\frac12e^{\psi_h}F'_h,
 \qquad \psi_h=\psi(r_h),\quad F'_h=F'(r_h).
 \label{eq:kappa}
\end{equation}
Writing $x=r-r_h$ gives $F=F'_h x+O(x^2)$. On the exterior side, the proper distance is $\rho=2\sqrt{x/F'_h}+O(x^{3/2})$. After $t=-\ii\tau$, the leading Euclidean metric is
\begin{equation}
 \dd s_E^2=\dd\rho^2+\kappa^2\rho^2\dd\tau^2+r_h^2\dd\Omega^2+cdots.
 \label{eq:rindler}
\end{equation}
Regularity of the $(\rho,\tau)$ plane requires
\begin{equation}
 \beta_H=\frac{2\pi}{\kappa},\qquad \Htemp=\frac{\kappa}{2\pi}.
 \label{eq:temperature}
\end{equation}
The interpretation of this period as a temperature uses the usual semiclassical thermal framework. No Einstein equation is needed to establish the local conical-regularity condition for the given metric.

\subsection{Residue and its normalization}
Complexify the local radial coordinate, $z=r-r_h$, and define the one-form
\begin{equation}
 \omega_T=\frac{e^{-\psi(r_h+z)}}{F(r_h+z)}\dd z.
 \label{eq:omega}
\end{equation}
It is the differential of the tortoise coordinate: radial null propagation obeys $\dd t=\pm\omega_T$ on the real section. This supplies a geometrical reason for choosing the differential.

\begin{proposition}\label{prop:residue}
Under the assumptions \eqref{eq:horizon}, $\omega_T$ has a simple pole at $z=0$, with
\begin{equation}
 \mathcal{R}_h\equiv\Res_{z=0}\omega_T
 =\frac{e^{-\psi_h}}{F'_h}=\frac{1}{2\kappa}.
 \label{eq:residue}
\end{equation}
With the Euclidean normalization \eqref{eq:temperature},
\begin{equation}
 \boxed{\beta_H=4\pi\mathcal{R}_h,\qquad
 \Htemp=\frac{1}{4\pi\mathcal{R}_h}.}
 \label{eq:restemp}
\end{equation}
\end{proposition}
\begin{proof}
The convergent local expansions of $F$ and $e^{-\psi}$ give
\begin{equation}
 \omega_T=\frac{e^{-\psi_h}}{F'_h}
 \left[\frac{1}{z}-\left(\psi'_h+\frac{F''_h}{2F'_h}\right)+O(z)\right]\dd z.
 \label{eq:laurent}
\end{equation}
The coefficient of $\dd z/z$ is Eq.~\eqref{eq:residue}; inserting it into Eq.~\eqref{eq:temperature} proves Eq.~\eqref{eq:restemp}.
\end{proof}

For a positively oriented contour $C$ that surrounds only this horizon,
\begin{equation}
 \mathcal{R}_h=\frac{1}{2\pi\ii}\oint_C\omega_T.
\end{equation}
The $4\pi$ in Eq.~\eqref{eq:restemp} is fixed by Euclidean regularity, not by the contour integral alone. In particular, the logarithmic monodromy of a radial tortoise coordinate should not be identified directly with a full Euclidean thermal period.

The residue belongs to a differential and is invariant under a locally invertible holomorphic radial change $z=h(w)$, with $h(0)=0$ and $h'(0)\ne0$. It is not the invariant coefficient of an arbitrarily redefined scalar function. Ramified maps, such as $z=w^2$, are not locally invertible at the pole and change the contour covering. They require a separate winding-number accounting.

Time normalization also matters. Under $\chi\mapsto a\chi$ with $a>0$, one has $\kappa\mapsto a\kappa$, $\Htemp\mapsto a\Htemp$, and $\mathcal R_h\mapsto\mathcal R_h/a$ when the adapted time is rescaled consistently. The temperature is therefore not determined by an unnormalized radial function alone. For a horizon with negative signed surface gravity one uses $\Htemp=|\kappa|/(2\pi)$; the present propositions use the outer-horizon sign convention in Eq.~\eqref{eq:horizon}.

\section{Noether charge and a contour representation of entropy}\label{sec:charges}
\subsection{Wald entropy and the horizon charge}
The Einstein--Hilbert scalar Lagrangian is
\begin{equation}
 L=\frac{R-2\Lambda}{16\pi G},\qquad
 \frac{\partial L}{\partial R_{abcd}}=
 \frac{g^{ac}g^{bd}-g^{ad}g^{bc}}{32\pi G}.
\end{equation}
For a bifurcation surface $B$ with binormal $\epsilon_{ab}\epsilon^{ab}=-2$, the Wald entropy is \cite{Wald1993,IyerWald1994}
\begin{equation}
 \SW=-2\pi\int_B\frac{\partial L}{\partial R_{abcd}}
 \epsilon_{ab}\epsilon_{cd}\dd A=\frac{A_H}{4G}.
 \label{eq:wald}
\end{equation}
The contraction of the two metric products with the binormals is $-4$, fixing the factor in the last equality. Thus the area-law normalization follows from the specified gravitational Lagrangian. It is not generated by a residue without gravitational input.

The gravitational Noether charge two-form is
\begin{equation}
 (\bm Q_\chi)_{ab}=-\frac{1}{16\pi G}\epsilon_{abcd}\nabla^c\chi^d.
 \label{eq:noether}
\end{equation}
At $B$, $\chi^a=0$ and $\nabla_a\chi_b=\kappa\epsilon_{ab}$. With the standard orientation yielding positive entropy,
\begin{equation}
 \QH\equiv\int_B\bm Q_\chi=\frac{\kappa A_H}{8\pi G}
 =\Htemp\SW.
 \label{eq:QH}
\end{equation}
For stationary horizons this gives the familiar charge--entropy relation. The identity $\nabla_a\chi_b=\kappa\epsilon_{ab}$ is being used on the bifurcation surface; it should not be applied indiscriminately to every horizon cross-section without examining additional terms.

\subsection{An explicitly normalized entropy differential}
Combining Eqs.~\eqref{eq:residue} and \eqref{eq:QH} gives
\begin{equation}
 \boxed{\SW=4\pi\QH\mathcal{R}_h.}
 \label{eq:product}
\end{equation}
For a fixed stationary solution, define
\begin{equation}
 \omega_S=4\pi\QH\,\omega_T.
 \label{eq:entropyform}
\end{equation}
Here $\QH$ is the horizon charge, held constant with respect to the auxiliary complex radial coordinate. It follows that
\begin{equation}
 \Res_{r=r_h}\omega_S=\SW,
 \qquad \SW=\frac{1}{2\pi\ii}\oint_C\omega_S.
 \label{eq:entropyresidue}
\end{equation}
Equation~\eqref{eq:entropyresidue} is a contour representation of an independently normalized entropy. It is not an equality between $\Res(e^{\cS}\dd z)$ and a Noether charge. Its two inputs are distinguishable: $\omega_T$ measures local radial redshift, whereas $\QH$ carries the gravitational coupling and the horizon area.

This distinction also makes the dimensions transparent. Since $[\omega_T]=L$ and $[\QH]=L^{-1}$, the residue of $\omega_S$ is dimensionless, as required for entropy. Under a constant rescaling of the Killing field, $\QH\mapsto a\QH$ and $\mathcal R_h\mapsto\mathcal R_h/a$, so their product is unchanged.

\subsection{Komar mass and the Smarr relation}
The Komar mass associated with a time translation has twice the gravitational Noether-charge normalization,
\begin{equation}
 \MK[\Sigma]=2\int_\Sigma\bm Q_\chi.
 \label{eq:komar}
\end{equation}
For a static, asymptotically flat, Ricci-flat exterior with $\Lambda=0$, the Killing identity and Stokes' theorem equate the relevant horizon and asymptotic Komar integrals. With the standard asymptotic time normalization,
\begin{equation}
 M_{\mathrm{ADM}}=\MK[\infty]=\MK[B]
 =\frac{\kappa A_H}{4\pi G}=2\Htemp\SW.
 \label{eq:smarr}
\end{equation}
The Noether charge $\QH$ is therefore $M_{\mathrm{ADM}}/2$ in this setting, not the full mass.

The Smarr relation is a scaling relation and must be distinguished from the differential first law \cite{Smarr1973,Bardeen1973}. For Schwarzschild, $\delta M=\Htemp\delta\SW$, while differentiating $M=2\Htemp\SW$ gives
\begin{equation}
 \delta M=2\Htemp\delta\SW+2\SW\delta\Htemp.
 \label{eq:dsmarr}
\end{equation}
The second term cannot be omitted. For a four-dimensional asymptotically flat Kerr--Newman solution, with the standard electromagnetic normalization and $\Lambda=0$,
\begin{align}
 M&=2\Htemp\SW+2\Omega_H J+\Phi_H Q,\\
 \delta M&=\Htemp\delta\SW+\Omega_H\delta J+\Phi_H\delta Q.
\end{align}
Matter, rotation, or a cosmological constant changes the charge balance and may add work or volume terms. Equality of unmodified horizon and asymptotic gravitational charges is not automatic in those cases.

\section{Explicit checks and counterexamples}\label{sec:tests}
\subsection{Schwarzschild geometry}
Take
\begin{equation}
 F(r)=1-\frac{2GM}{r},\qquad\psi=0,\qquad r_h=2GM.
\end{equation}
The exact differential near the horizon is
\begin{equation}
 \omega_T=\frac{r\dd r}{r-r_h}
 =\left(1+\frac{r_h}{z}\right)\dd z,
 \qquad \mathcal R_h=r_h=2GM.
\end{equation}
Thus
\begin{equation}
 \kappa=\frac{1}{4GM},\quad
 \Htemp=\frac{1}{8\pi GM},\quad
 A_H=16\pi G^2M^2,\quad
 \SW=4\pi GM^2.
\end{equation}
The charge is $\QH=M/2$, and Eq.~\eqref{eq:product} gives
\begin{equation}
 4\pi\QH\mathcal R_h=4\pi\frac{M}{2}(2GM)=4\pi GM^2=\SW.
\end{equation}
Direct differentiation also yields $\Htemp\dd\SW/\dd M=1$, verifying the first law rather than inferring it from pole order.

The derivative-square density in Eq.~\eqref{eq:trial} illustrates why exponentiated entropy behaves differently. For $\chi=\partial_t$ and $\psi=0$,
\begin{equation}
 \nabla_r\chi_t=-\frac12F',\qquad
 \nabla_t\chi_r=\frac12F',\qquad
 \nabla_a\chi_b\nabla^a\chi^b=-\frac12(F')^2.
 \label{eq:kinetic}
\end{equation}
This Lorentzian scalar is finite at a nonextremal Schwarzschild horizon. A finite-time radial integral over a thin horizon shell therefore has a regular local expansion. Its exponential need not have any pole at all, although the background satisfies the vacuum Einstein equation.

\subsection{Identical horizon data without vacuum dynamics}\label{sec:offshell}
Consider the asymptotically flat family
\begin{equation}
 F_\eps(r)=1-\frac{r_h}{r}
 +\eps\frac{r_h^2(r-r_h)^2}{r^4},\qquad\psi=0,
 \qquad r_h>0,\quad\eps>0.
 \label{eq:offshell}
\end{equation}
For $r>r_h$ both terms beyond the horizon factor give $F_\eps>0$. The horizon remains simple and
\begin{equation}
 F_\eps(r_h)=0,\qquad F'_\eps(r_h)=\frac{1}{r_h},\qquad
 \mathcal R_h=r_h,\qquad\kappa=\frac{1}{2r_h}.
\end{equation}
Its area is $4\pi r_h^2$ and its ADM mass is $r_h/(2G)$, because the deformation begins at order $r^{-2}$ at infinity. These quantities coincide with those of Schwarzschild at the same $r_h$.

Nevertheless, the mixed Einstein tensor for the $\psi=0$ metric obeys
\begin{equation}
 G^t{}_t=G^r{}_r=\frac{rF'+F-1}{r^2},\qquad
 G^\theta{}_\theta=G^\varphi{}_\varphi=\frac{F''}{2}+\frac{F'}{r}.
 \label{eq:einsteinspherical}
\end{equation}
Substitution gives
\begin{align}
 G^t{}_t&=-\eps\frac{r_h^2(r-r_h)(r-3r_h)}{r^6},\label{eq:residualtt}\\
 G^\theta{}_\theta&=\eps\frac{r_h^2(r^2-6r_hr+6r_h^2)}{r^6}.
 \label{eq:residualtheta}
\end{align}
For example, $G^t{}_t(2r_h)=\eps/(64r_h^2)$ and the limiting angular component at the horizon is $G^\theta{}_\theta(r_h)=\eps/r_h^2$. The metric therefore violates the vacuum equation with $T_{ab}=0$ and $\Lambda=0$ for every $\eps>0$.

This is a geometric counterexample to sufficiency: the complete leading horizon residue, not just its pole order, can agree with Schwarzschild while the vacuum field equation fails. The deformation may be assigned an effective stress tensor $G_{ab}/(8\pi G)$, but that changes the prescribed matter problem. No claim is made that it satisfies a particular matter action or energy condition. It also shows that endpoint mass and horizon data alone do not reconstruct the local exterior dynamics.

\subsection{Extremal horizons and higher-order poles}
For a charged static geometry written with a geometrical charge parameter $q$ of dimension length,
\begin{equation}
 F(r)=1-\frac{2GM}{r}+\frac{q^2}{r^2}
 =\frac{(r-r_+)(r-r_-)}{r^2},\quad
 r_\pm=GM\pm\sqrt{G^2M^2-q^2},
\end{equation}
the nonextremal outer-horizon residue is
\begin{equation}
 \mathcal R_+=\frac{r_+^2}{r_+-r_-},\qquad
 \kappa_+=\frac{r_+-r_-}{2r_+^2}.
 \label{eq:RNresidue}
\end{equation}
In the extremal limit $q=GM=r_h$, the zeros merge and
\begin{equation}
 F(r)=\frac{(r-r_h)^2}{r^2},\qquad
 \omega_T=\left(\frac{r_h^2}{z^2}+\frac{2r_h}{z}+1\right)\dd z.
 \label{eq:extremal}
\end{equation}
The double pole occurs in an exact Einstein--Maxwell solution. The residue $2r_h$ remains finite, even though $\kappa=0$. Substituting this residue into Eq.~\eqref{eq:restemp} would incorrectly predict a nonzero temperature. There is no contradiction: the simple-zero hypothesis has failed, and the nonextremal Euclidean polar-coordinate derivation does not apply. In particular, taking an extremal limit and taking the residue at a merged pole are different operations.

\subsection{What the checks establish}
Table~\ref{tab:checks} summarizes the scope of the construction.We check the explicit algebraic identities, Schwarzschild thermodynamic factors, a regular radial reparametrization, the nonextremal charged case, the extremal Laurent form, and the vacuum residual of Eq.~\eqref{eq:offshell}. 

\begin{table}[tb]
\centering\small
\caption{Logical role of the principal calculations.}\label{tab:checks}
\begin{tabularx}{\textwidth}{@{}p{0.25\textwidth}XX@{}}
\toprule
Setting & Verified statement & Limitation \\
\midrule
Null-vector variation & Stationarity for all null directions gives Eq.~\eqref{eq:einstein}. & Requires the prescribed functional and conserved stress tensor. \\
Nondegenerate static horizon & $\mathcal R_h=1/(2\kappa)$. & Uses analyticity and fixed time normalization. \\
Schwarzschild & $\SW=4\pi\QH\mathcal R_h$ and $M=2\Htemp\SW$. & Uses the Einstein--Hilbert charge normalization. \\
Deformed exterior & Horizon residue and asymptotic mass are unchanged. & Vacuum field equations fail for $\eps>0$. \\
Extremal charged horizon & A double pole coexists with exact field equations. & The simple-pole temperature formula is inapplicable. \\
\bottomrule
\end{tabularx}
\end{table}

\section{Physical interpretation and further tests}\label{sec:discussion}
\subsection{A controlled thermodynamic interpretation}
The variational and analytic parts answer different questions. The antisymmetrized functional encodes dynamics through a curvature commutator and a requirement on every null direction. The temperature differential encodes the leading redshift at a specified nondegenerate horizon. The Wald construction supplies the entropy appropriate to a specified Lagrangian. Their compatibility is useful, but it does not imply that the residue alone contains all of the bulk equations.

Within this separation, the interpretation of geometry through thermodynamics remains substantive. The Einstein equation can be expressed as null-direction entropy stationarity, and the temperature of a static horizon can be extracted from a local analytic coefficient. The contour form in Eq.~\eqref{eq:entropyresidue} packages these data into one expression while making its normalization explicit. It does not predict a correction to the Hawking temperature or the area law in Einstein gravity.

The construction is also distinct from a proof that a gravitational singularity causes attraction. A meromorphic pole in $\omega_T$ occurs at a regular nonextremal horizon; its physical role is tied to redshift and the thermal period. A curvature singularity or incomplete geodesic is a different diagnostic. A proposal connecting such a breakdown of classical geometry to a thermodynamic singularity must identify the relevant ensemble, potential, and limiting process. Neither a divergence nor a zero temperature by itself specifies a force law.

\subsection{Pole order, dissipation, and global equilibrium}
Higher-order poles cannot be identified universally with entropy production. The extremal example provides a direct obstruction. Conversely, a simple pole does not establish global thermal equilibrium: several horizons can have different surface gravities, and one Euclidean period need not regularize all of them simultaneously. A thermal state and its boundary conditions must be specified separately.

Non-equilibrium thermodynamics requires an entropy-balance equation, schematically
\begin{equation}
 \dd S=\frac{\delta Q}{T}+\dd_iS,
 \qquad \dd_iS\ge0,
\end{equation}
together with a physical model for $\dd_iS$. Such terms have been studied in spacetime thermodynamics \cite{Eling2006}. A Laurent coefficient acquires a dissipative meaning only after a derivation connects it to the fluxes or transport coefficients in that balance law. No such connection follows from pole classification alone.

\subsection{Comparison with topological temperature methods}
The Robson--Villari--Biancalana approach relates the temperature to the topology and curvature of a Euclidean reduction \cite{Robson2019}. Here the radial differential is fixed by null propagation, and its normalization is checked against local Euclidean regularity. We do not derive its contour from a Gauss--Bonnet integral or claim equivalence to all global topological prescriptions. Agreement for a Schwarzschild horizon is a consistency check, not a new temperature prediction.

An extension of a topological or residue method should state which additional input it changes: the integration constant, horizon selection, time normalization, global Euclidean structure, or underlying field equation. A proposed new physical effect would then require a quantity that differs from the baseline calculation after these conventions have been matched.

\subsection{Extensions and falsifiable requirements}
The null-vector construction extends to suitable divergence-free curvature tensors in Lovelock theories \cite{PadmanabhanParanjape2007}. It cannot be transferred unchanged to a generic $f(R)$ theory, where derivatives of $\partial L/\partial R_{abcd}$ contribute to the equations. The horizon entropy and Noether charge must also be recalculated for the chosen Lagrangian. Equation~\eqref{eq:product} can retain its form in a static setting when the relevant stationary charge relation holds, but this does not supply the modified dynamics.

A stronger theory based on an entropy-generating analytic object would need to specify that object independently of the desired residue, derive its transformation properties, fix its additive and multiplicative normalizations, and recover the field equations beyond a single solution. It must distinguish the family \eqref{eq:offshell} from Schwarzschild when vacuum dynamics is required, and explain why the extremal limit invalidates the simple-pole formula. These are concrete tests, rather than interpretive analogies. Establishing a correction to a temperature, entropy, or response function would require additional dynamical or microscopic input not supplied here.

\section{Conclusions}
We have formulated a consistent relation between null-vector entropy stationarity, horizon residues, and gravitational charges. The Einstein equation follows from the established antisymmetrized entropy functional once the null constraint and stress conservation are treated explicitly. A bare derivative-square functional does not provide the same derivation.

For a locally analytic static spherical geometry with a nondegenerate outer horizon, the radial propagation differential has residue $\mathcal R_h=1/(2\kappa)$. Euclidean regularity fixes $\Htemp=(4\pi\mathcal R_h)^{-1}$. The independently normalized horizon charge then gives $\SW=4\pi\QH\mathcal R_h$, reproducing the standard Schwarzschild relations and separating the Komar mass from the Noether charge.

The counterexamples delimit the result. An asymptotically flat deformation preserves the Schwarzschild horizon residue, area, surface gravity, and mass but fails the vacuum equation. An extremal charged solution has a higher-order pole. Consequently, the absence of higher-order poles is neither a universal criterion for gravitational dynamics nor a general condition of vanishing entropy production. The proposed contour representation is a controlled reformulation of stationary horizon thermodynamics; a new microscopic explanation or physical correction remains a question for further work.

\appendix
\section{Connection with local horizon entropy balance}\label{app:jacobson}
For completeness, the null projection in Eq.~\eqref{eq:nullprojection} can be checked by the local horizon argument \cite{Jacobson1995}. Let $k^a$ be affinely parametrized by $\lambda$, with $\lambda=0$ at a chosen point and local expansion and shear zero there. To leading order, Raychaudhuri's equation gives
\begin{equation}
 \frac{\dd\theta}{\dd\lambda}=-R_{ab}k^ak^b+O(\lambda),
 \qquad\theta=-\lambda R_{ab}k^ak^b+O(\lambda^2).
\end{equation}
For an approximate boost generator $\chi^a=-\kappa\lambda k^a$ on the past horizon segment, the heat flux and entropy change are
\begin{align}
 \delta Q&=-\kappa\int\lambda T_{ab}k^ak^b\,\dd\lambda\dd A,\\
 \delta S&=-\eta\int\lambda R_{ab}k^ak^b\,\dd\lambda\dd A.
\end{align}
Imposing $\delta Q=T\delta S$ with $T=\kappa/(2\pi)$ and $\eta=1/(4G)$ yields
\begin{equation}
 R_{ab}k^ak^b=8\pi GT_{ab}k^ak^b.
\end{equation}
The tensor lemma and conservation then recover Eq.~\eqref{eq:einstein}. This calculation provides a separate thermodynamic check of the coefficient, but assumes the area entropy density and the local Unruh temperature. It does not derive those inputs or establish an equivalence with a Laurent-pole condition.

\section{Spherical charge and an explicit coordinate check}
For $\psi=0$ in Eq.~\eqref{eq:metric}, direct evaluation of the Komar integral gives
\begin{equation}
 \MK(r)=\frac{r^2F'(r)}{2G}.
\end{equation}
Schwarzschild therefore has $\MK(r)=M$. For the deformed geometry in Eq.~\eqref{eq:offshell},
\begin{equation}
 \MK(r)=\frac{r_h}{2G}+\frac{\eps}{G}
 \left(-\frac{r_h^2}{r}+\frac{3r_h^3}{r^2}-\frac{2r_h^4}{r^3}\right).
\end{equation}
The deformation vanishes at the horizon and at infinity but not on generic intermediate spheres. Equal endpoint values therefore do not establish local vacuum charge conservation.

For the radial coordinate change $z=aw+bw^2$ with $a\ne0$, the singular part of a simple-pole differential transforms as
\begin{equation}
 \mathcal R_h\frac{\dd z}{z}
 =\mathcal R_h\frac{a+2bw}{aw+bw^2}\dd w
 =\mathcal R_h\left(\frac1w+\frac ba+O(w)\right)\dd w.
\end{equation}
The residue is unchanged. If the Jacobian were omitted and only the scalar coefficient were re-expanded, an incorrect coordinate dependence would result.

\end{document}